\documentclass[10pt,a4paper]{article}
\usepackage[T1]{fontenc}
\usepackage{url}

\usepackage[linesnumbered,ruled,vlined,noresetcount]{algorithm2e}
\DontPrintSemicolon
\SetArgSty{textnormal}
\SetKwProg{Fn}{function}{:}{end}
\SetKwInput{Colors}{Colors}
\SetKwInput{Notation}{Notation}
\SetKwInput{RuleFunction}{Rule Function}

\usepackage{multirow}
\usepackage{color}

\usepackage{authblk}
\usepackage{geometry}
\usepackage{amsmath,amssymb,amsthm}
\usepackage{graphicx}

\newtheorem{lemma}{Lemma}
\newtheorem{theorem}{Theorem}

\theoremstyle{definition}

\newtheorem{definition}{Definition}
\newtheorem{note}{Note}

\newcommand{\bnames}{Böckenhauer, Frei, Unger, and Wehner}

\newcommand{\gint}{\gamma_{\mathrm{init}}}

\newcommand{\ie}{{\em i.e.},~}

\newcommand{\protocol}[1]{\textsc{#1}}

\newcommand{\colorName}{\mathsf}
\newcommand{\cneigh}{\colorName{neigh}}
\newcommand{\cpath}{\colorName{path}}
\newcommand{\cinit}{\colorName{init}}
\newcommand{\cfin}{\colorName{fin}}
\newcommand{\cheadd}{\colorName{head}}
\newcommand{\cheadone}{\colorName{head1}}
\newcommand{\cheadtwo}{\colorName{head2}}
\newcommand{\cstay}{\colorName{stay}}
\newcommand{\cstop}{\colorName{stop}}

\newcommand{\tr}{\mathbf{true}}
\newcommand{\fl}{\mathbf{false}}

\newcommand{\mif}{\mathbf{if~}}
\newcommand{\melif}{\mathbf{else~if~}}
\newcommand{\moth}{\mathbf{otherwise}}

\newcommand{\ssix}{\protocol{Gen6}}
\newcommand{\asix}{\calA_{\ssix}}
\newcommand{\csix}{C_{\ssix}} 
\newcommand{\xisix}{\xi_{\ssix}} 
\newcommand{\sfive}{\protocol{PF5}}
\newcommand{\afive}{\calA_{\sfive}}
\newcommand{\cfive}{C_{\sfive}} 
\newcommand{\xifive}{\xi_{\sfive}}
\newcommand{\stf}{\protocol{TF5}}
\newcommand{\atf}{\calA_{\stf}}
\newcommand{\ctf}{C_{\stf}}
\newcommand{\xitf}{\xi_{\stf}}

\newcommand{\upf}{N(u_k)\setminus (F \cup \{u_1,u_2,\dots,u_k\}\cup \bigcup_{i=1}^{k-1} N(u_i))}
\newcommand{\upfg}[1]{N(u_k)\setminus (F(#1) \cup \{u_1,u_2,\dots,u_k\}\cup \bigcup_{i=1}^{k-1} N(u_i))}

\newcommand{\calA}{\mathcal{A}}

\newcommand{\calM}{\mathcal{M}}
\newcommand{\calE}{\mathcal{E}}

\newcommand{\rlabel}[1]{\textbf{Rule #1}}
\newcommand{\rlabels}[1]{\textbf{Rules #1}}

\newcommand{\vin}{v_{\mathrm{in}}}
\newcommand{\vout}{v_{\mathrm{out}}}

\title{Recolorable Graph Exploration by an Oblivious Agent with Fewer Colors}
\date{}

\author[1]{Shota Takahashi}
\author[2]{Haruki Kanaya}
\author[3]{Shoma Hiraoka}
\author[2]{Ryota Eguchi}
\author[1]{Yuichi Sudo\thanks{Corresponding Author: sudo@hosei.ac.jp}}

\affil[1]{Hosei University, Tokyo, Japan}
\affil[2]{Nara Institute of Science and Technology, Nara, Japan}
\affil[3]{Osaka University, Osaka, Japan}

\begin{document}
\maketitle
\begin{abstract}
Recently, Böckenhauer, Frei, Unger, and Wehner (SIROCCO 2023) introduced a novel variant of the graph exploration problem, in which a single memoryless agent must visit all nodes of an unknown, undirected, and connected graph and return to its starting node. Unlike the standard model for mobile agents, edges are not labeled with port numbers. Instead, the agent can color its current node and observe the color of each neighboring node. To move, it specifies a target color and then moves to an adversarially chosen neighbor of that color.

They studied the minimum number of colors required for successful exploration and proposed an elegant algorithm that explores any graph using only eight colors. In this paper, we present a new graph exploration algorithm that requires only six colors. Furthermore, we prove that five colors suffice for a restricted class of graphs, called $\varphi$-free graphs, which includes all graphs with maximum degree at most three and all cactus graphs. We also prove that five colors suffice for all graphs that do not contain a cycle of length three.
\end{abstract}


\section{Introduction}
\label{sec:introduction}
Autonomous mobile entities that migrate from node to node on a graph---commonly referred to as \emph{mobile agents} or simply \emph{agents}---have been extensively studied in the theoretical distributed computing community. Numerous fundamental problems involving these agents have been investigated in the literature, including graph exploration~\cite{AKLLR79,PDD+96,PP99,Koucky02,BFRSV02,Reingold08,CFIKP08,Ilcinkas08,HZZT12,SBN+15,DHK19,SOK26,BGKP22,PP24,DDGM25}, gathering~\cite{OKK+14,SNO+19}
(called rendezvous if exactly two agents are considered), dispersion~\cite{KA19,KMS18,KMS19,AM18,SSN+24,KKA+25,PKK+26}, uniform deployment \cite{SMO+16,SSNK20},  black hole search~\cite{DFPS06,DFPS07,DSS+08}, and others.

Among these, graph exploration (performed by a single agent) is arguably one of the most fundamental problems, since exploration algorithms often serve as building blocks for solving many other problems. The objective is to enable the agent to visit every node in a graph, return to its starting node, and then terminate. This problem has been studied under various settings: the nodes may be \emph{labeled} or \emph{anonymous}; the agent may have access to the local memories of the nodes (often referred to as \emph{whiteboards}); the agent may have its own memory or be \emph{oblivious}; and the initial configuration may be fixed or arbitrary (in which case the solution is self-stabilizing), among others.

All these settings for graph exploration share a common characteristic: the edges are locally labeled with \emph{port numbers}. Typically, in a graph $G=(V,E)$, the edges incident to a node $v\in V$ are assigned distinct port numbers $0,1,\dots,\delta_v-1$, where $\delta_v$ is the degree of $v$. The agent uses these port numbers to distinguish between the edges and decide which neighbor to move to. Moreover, most studies in the literature assume that the agent can access the \emph{incoming port} information: when the agent moves from a node $u$ to $v$, it learns the port number assigned to the edge $\{u,v\}$ at the destination $v$, whereas the well-known rotor-router algorithm \cite{PDD+96} and some other algorithms \cite{BGKP22} do not use this information.

In 2023, Böckenhauer, Frei, Unger, and Wehner~\cite{BFUW23} introduced a novel model for mobile agents in which \emph{port numbers are unavailable}. Moreover, they assume that the nodes are anonymous and that the agent is oblivious, so additional assumptions are required to make graph exploration feasible. To overcome these limitations, they introduced \emph{node coloring}. Specifically, at each time step, the agent at a node $v$ can observe the color of $v$ as well as the colors of its neighboring nodes. Based on this information, the agent updates the color of $v$ and computes a \emph{destination color} $c$, after which it moves to a neighbor with color $c$. If two or more neighbors of $v$ are colored $c$, the destination is selected adversarially among them.
At the start of exploration, all nodes are assigned a common initial color.

The authors of~\cite{BFUW23} thoroughly investigated the minimum number of colors required to solve this new variant of the graph exploration problem. It is worth mentioning that they measure color complexity by excluding the initial color. In contrast, we include the initial color in our count---a convention we believe improves clarity. Consequently, whenever we refer to results from~\cite{BFUW23}, our reported color counts are always one greater.

If the agent is restricted to changing the color of each node at most once, they prove that $n$ colors are both necessary and sufficient for arbitrary simple, undirected, and connected graphs, while four colors are necessary and sufficient for trees. (Throughout the remainder of this paper, we assume that the graph is simple, undirected, and connected, and thus omit this assumption from subsequent descriptions.) They further show that $n-1$ colors remain necessary even for bipartite graphs of size $n=|V|$. Moreover, they prove that if the exact network size $n=|V|$ is known, then $n-2$ and $n-1$ colors are necessary and sufficient, respectively, for exploring arbitrary graphs.

In the model where the agent is allowed to recolor each node an unlimited number of times, which we refer to as the recolorable model, the color complexity decreases significantly.
Specifically, they prove that eight colors are sufficient to explore arbitrary graphs, although no lower bound is provided.

We are not aware of any immediate practical application of this exploration setting. However, we believe that this variant is of significant theoretical interest. In the absence of port numbers, most commonly used techniques for graph exploration and related problems are no longer applicable. Consequently, entirely new methods are required to minimize the color complexity. In this regard, the model introduced by~\cite{BFUW23} may open new avenues for research in the field of mobile agents.

\subsection{Our Contribution}
In this paper, we address the aforementioned variant of graph exploration introduced by~\cite{BFUW23}. In particular, within the recolorable model, we present an exploration algorithm for arbitrary graphs that uses only six colors, improving on the eight-color bound of~\cite{BFUW23} by two. Moreover, for triangle-free graphs---those containing no cycle of length 3---this algorithm can be adapted to use only five colors. We also prove that five colors suffice for another restricted class of graphs, which we call \emph{$\varphi$-free} graphs (see Section~\ref{sec:afive}); this class includes all graphs of maximum degree 3 and all cactus graphs.

Similarly to~\cite{BFUW23}, we do not provide any non-trivial lower bound for exploration on arbitrary graphs in the recolorable model. Therefore, the main open question is whether six colors are indeed necessary or if this number can be further reduced. Another promising direction for future research is to allow the agent to have its own internal memory. Even a constant number of bits---or even just one bit---might significantly help to reduce the required number of colors.

\begin{note}
We adopt an interpretation of the model in~\cite{BFUW23}
that allows the agent to recolor a node from a non-initial
color back to the initial color.
The formal definition in that paper leaves room for
interpretation regarding this operation.
We thank one of the authors for clarifying in personal
communication that their intended model does not permit
this operation.
Under this restriction, our algorithms for arbitrary graphs
and triangle-free graphs can be adapted using one additional
color, increasing their color counts from six to seven and
from five to six, respectively.
Our algorithm for $\varphi$-free graphs remains unchanged
and still uses five colors, since it does not require
recoloring nodes back to the initial color.
\end{note}

\subsection{Technical Overview}
The eight-color algorithm of B\"ocken\-hauer, Frei, Unger, and
Wehner~\cite{BFUW23} follows a breadth-first search approach.
It combines three distance labels, encoding distances from the
starting node modulo three, with two states used to track progress
through successive exploration phases.
These combinations require six colors; an initial color and a
color marking nodes as \emph{finished}, meaning that they need
not be visited again, bring the total to eight.

Our algorithm instead simulates a modified depth-first search,
which we call \emph{semi-DFS}.
A difficulty with ordinary DFS is that the current node may
have several neighbors on the return path, making it difficult
to identify its predecessor using only a few colors.
We avoid this ambiguity by maintaining the return path as an
induced path, with no edges between nonconsecutive nodes.
To preserve this property, we extend the path only to an unfinished
node outside the path that has no neighbor on it other than its
current endpoint.
Consequently, whenever the path contains more than one node,
the endpoint has a unique neighbor among the earlier nodes on
the path, allowing unambiguous backtracking without port numbers.
We show that this restricted path extension still allows
every node to be visited, and implement the procedure using
six colors to represent the return path, finished nodes,
and temporary states used when testing possible extensions.
For triangle-free graphs, we further show that the two colors
used to mark the head of the return path can be merged into
a single color, reducing the color complexity to five.

For $\varphi$-free graphs, we use a different approach inspired by
the distance labels in the eight-color algorithm.
Instead of computing distances from the starting node, we use
three colors $d_0,d_1,d_2$ to maintain a directed acyclic graph:
edges between nodes colored $d_i$ and $d_{i+1}$ are directed
from the former to the latter, with indices taken modulo three.
We assign these colors so that every node carrying one of them
remains reachable from the starting node in this directed graph.
A difficulty arises when a newly visited node has neighbors
of all three colors.
The $\varphi$-free condition guarantees that such a node has
degree exactly three, allowing it to be marked as finished
without disconnecting the remaining unfinished nodes.
This structure guides forward moves and backtracking while
ensuring that no unvisited region is cut off.
Together with the initial and finished colors, the three
directional colors yield a five-color exploration algorithm.

\subsection{Related Work}
Graph exploration has been studied under various assumptions about
the information and memory available to an agent.

When nodes have unique identifiers, a standard depth-first traversal
explores all nodes and edges of a connected undirected graph in at
most $2m$ moves, where $m$ is the number of edges.
Panaite and Pelc~\cite{PP99} reduced this bound to $m+3n$, where $n$
is the number of nodes.
Since traversing all edges requires at least $m$ moves, this reduces
the additive overhead above the lower bound from $m$ to $O(n)$.
Sudo, Baba, Nakamura, Ooshita, Kakugawa, and Masuzawa~\cite{SBN+15}
used whiteboards (writable node memory) to reduce the internal memory
required by these algorithms and analyzed tradeoffs between
internal memory and whiteboard size.

Universal sequences allow an agent to explore anonymous graphs
without whiteboards.
Using the probabilistic method, Aleliunas, Karp, Lipton,
Lov\'asz, and Rackoff~\cite{AKLLR79} proved the existence of
polynomial-length \emph{universal traversal sequences} (UTSs).
For given bounds $N$ and $d$, a UTS is a fixed sequence of outgoing
port numbers that explores every connected undirected graph with
at most $N$ nodes and maximum degree at most $d$, regardless of
the starting node and local port numbering.%
\footnote{The original formulation assumes regular graphs,
but arbitrary graphs can be handled by adding virtual self-loops:
instructions specifying nonexistent ports are interpreted as
staying at the current node.
See Ta-Shma and Zwick~\cite{TZ14} for details.}
Unlike UTSs, which do not use incoming-port information,
\emph{universal exploration sequences} (UXSs), introduced by
Kouck\'y~\cite{Koucky02}, specify port offsets relative to the
incoming port, allowing backtracking.
Reingold~\cite{Reingold08} established deterministic
logarithmic-space undirected connectivity and gave a
logarithmic-space construction of polynomial-length UXSs.
This yields deterministic exploration in time polynomial in
a known size bound $N \ge n$, using $O(\log N)$ bits of
internal memory.

Writable node memory offers another approach.
Priezzhev, D.~Dhar, A.~Dhar, and Krishnamurthy~\cite{PDD+96}
introduced Eulerian walkers, now widely known as the
\emph{rotor-router} algorithm.
Each node maintains an outgoing-port pointer, which the agent
advances cyclically before traversing the corresponding edge.
The algorithm is self-stabilizing: it guarantees perpetual
exploration from arbitrary initial pointer values, without
internal memory or a known upper bound on $n$.
Sudo, Ooshita, and Kamei~\cite{SOK26} use internal memory and whiteboards to improve cover-time
bounds from arbitrary initial configurations.
Bojko, Gotfryd, Kowalski, and Paj\k{a}k~\cite{BGKP22} studied tree
exploration without knowledge of the incoming port, while allowing
the agent to distinguish outgoing ports.

Exploration has also been studied with short labels assigned to
nodes during preprocessing.
Cohen, Fraigniaud, Ilcinkas, Korman, and Peleg~\cite{CFIKP08} showed
that three distinct labels, encoded using two bits per node, allow
an agent with constant-size internal memory to explore any graph
in $O(m)$ time.
They also gave a scheme using one-bit labels and $O(\log \Delta)$
bits of internal memory, where $\Delta$ is the maximum degree.
Hu, Zhang, Zhang, and Tang~\cite{HZZT12} studied the tradeoff between
the proportion of black nodes in one-bit labeling schemes and
the agent's memory requirements.
A different preprocessing approach is to choose the port numbering:
Ilcinkas~\cite{Ilcinkas08} showed that a suitable assignment allows
a three-state automaton to explore any $n$-node graph periodically
with a period of $O(n)$.

Another line of research uses pebbles placed on nodes to help
the agent recognize them.
For anonymous strongly connected directed graphs,
Bender, Fern\'andez, Ron, Sahai, and Vadhan~\cite{BFRSV02} gave
a polynomial-time mapping algorithm using one movable pebble
when an upper bound on the graph size is known.
Without such a bound, they showed that $\Theta(\log\log n)$
pebbles are necessary and sufficient for polynomial-time mapping.
For undirected graphs, Disser, Hackfeld, and Klimm~\cite{DHK19}
proved that $\Theta(\log\log n)$ distinguishable pebbles are
necessary and sufficient for exploration by an agent with
constant-size internal memory.
Pattanayak and Pelc~\cite{PP24} studied how movable,
indistinguishable pebbles affect the exploration power of
deterministic memoryless automata.
Colored pebbles placed by an oracle before exploration have also
been studied.
Das, Dhar, Gorain, and Mahawar~\cite{DDGM25} considered collision-free
exploration by two identical agents and showed that three pebble
colors suffice for arbitrary graphs, whereas one color does not.

Graph exploration has also been studied in dynamic graphs, whose edge sets may change at discrete time steps.
Previous work has considered dynamic rings~\cite{IW18,DDFS20,GSOM20}, dynamic tori~\cite{GSO+21}, dynamic graphs with arbitrary underlying topologies~\cite{GFMS21,SKS+26}, and carrier graphs, where edge availability is determined by the periodic movements of carriers~\cite{FMS13}.


\section{Preliminaries}
\label{sec:pre}
Throughout this paper, the term \emph{graph} refers exclusively to a simple, undirected, and connected graph.

An algorithm is defined as a triple $\calA = (C, c_0, \xi)$, where $C$ is a set of colors, $c_0 \in C$ is the initial color, and
$$
\xi: C \times \calM(C) \to C \times (C \cup \{\cstay,\cstop\})
$$
is a rule function. Here, $\calM(C)$ denotes the set of all multisets over $C$, that is, all functions $m: C \to \mathbb{N}_0$, where $\mathbb{N}_0$ is the set of nonnegative integers.
When an agent executes $\calA$ on a graph $G=(V,E)$, the execution proceeds as follows. Let $v_t$ be the node where the agent is located at time $t=0,1,2,\dots$, and let $c(v,t)$ denote the color of the node $v\in V$ at time $t$. Initially, every node is assigned the color $c_0$, \ie $c(v,0)=c_0$ for all $v\in V$.
At each time step $t$, the agent observes the color of its current node $v_t$ and the colors of all its neighbors. More precisely, it obtains the current color $c(v_t,t)$ and the multiset $M=\bigsqcup_{u\in N(v_t)}\{c(u,t)\}$.
Here, the symbol $\sqcup$ denotes the sum (union) of multisets, and $N(v)$ represents the set of neighbors of a node $v\in V$. The agent then recolors $v_t$ with a new color $x$ and moves to a neighbor with color $y$, where
$$
(x,y)=\xi(c(v_t,t),M).
$$
Note that if more than one neighbor in $N(v_t)$ has color $y$, the agent moves to one of those nodes, with the choice made \emph{adversarially}. We have two exceptions:
\begin{itemize}
\item The agent stays at the current node $v_t$ if $y=\cstay$\footnotemark{}, and
\item The agent terminates its execution at $v_t$ if $y \neq \cstay$ and either $y = \cstop$ or $y \notin M$.
\end{itemize}
\footnotetext{
The stay option $\cstay$ is not included in~\cite{BFUW23}.
We introduce it only to simplify the description of our algorithms.
In these algorithms, every sequence of consecutive stay steps
ends in a move or termination.
Thus, stay steps can be eliminated by recursively replacing
$\xi(c,M)=(x,\cstay)$ with $\xi(c,M)=\xi(x,M)$,
yielding an equivalent algorithm without the stay option.
}
Even in these exceptional cases, the agent recolors $v_t$ with $x$.

A global state, or \emph{configuration}, is a triple $\gamma=(v,\psi,b)$, where $v\in V$ is the current node where the agent is located, $\psi:V\to C$ assigns a color to every node, and $b\in\{\fl,\tr\}$ indicates whether the algorithm has terminated. We write $\gamma\to\gamma'$ to denote that configuration $\gamma$ transitions to $\gamma'$ by one action of the agent (i.e., by updating the color and executing a movement or termination).
We denote by $\gint(v)$ the configuration in which all nodes are colored $c_0$, the agent is located at node $v$, and the agent has not terminated. A finite or infinite sequence $\calE=\gamma_0,\gamma_1,\dots$ of configurations is called an \emph{execution} of an algorithm $\calA$ starting at node $v$ if $\gamma_0=\gint(v)$ and $\gamma_t\to\gamma_{t+1}$ holds for every consecutive pair of configurations. We say that an execution $\calE$ is \emph{maximal} if it is infinite or if the agent has terminated in its last configuration.

\begin{definition}[Graph Exploration]
An algorithm $\calA$ is said to solve \emph{exploration} on a graph $G=(V,E)$ if, for any node $s \in V$,
the agent visits every node in $V$ and terminates at $s$ in every maximal execution of $\calA$ starting at $s$.
\end{definition}

Our goal is to minimize the number of colors used in graph exploration. For an algorithm $\calA = (C, c_0, \xi)$, we define the number of colors to be $|C|$. Note that in the original work by \bnames~\cite{BFUW23}, the number of colors is defined as $|C| - 1$, since the initial color $c_0$ is not counted. In this paper, we count the initial color $c_0$, as this definition may be clearer for most readers.

\section{Exploration on Arbitrary Graphs}
\label{sec:asix}
In this section, we show that six colors suffice for exploring arbitrary graphs. \bnames~\cite{BFUW23} presented a simple and elegant algorithm for graph exploration that employs eight colors. In their algorithm, the agent visits all nodes in a breadth-first search (BFS) manner and terminates at the starting node $s$. To reduce the number of colors, we adopt a different approach in which the agent explores the graph in a manner reminiscent of depth-first search (DFS), though its movement deviates slightly. We refer to this approach as the \emph{semi-DFS}. In the remainder of this section, we first explain the semi-DFS and then show how the agent can simulate it using only six colors.

\begin{figure}[t]
 \centering 
 \includegraphics[width=0.3\textwidth]{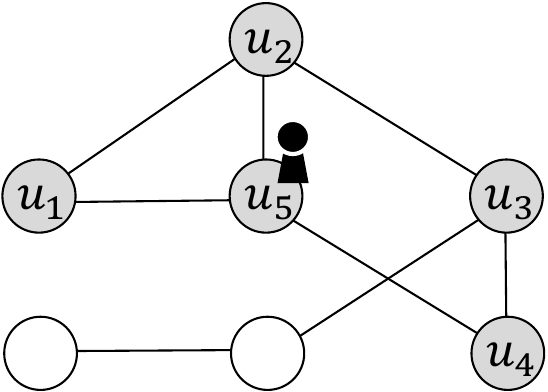}
 \caption{An example configuration illustrating why simulating DFS is difficult. The agent must backtrack from $u_5$ to $u_4$, but it cannot determine which neighbor of $u_5$ corresponds to $u_4$.} 
 \label{fig:path}
\end{figure}

\subsection{Semi-DFS}
A depth-first search (DFS) can be described in terms of a maintained path $P$ and a set $F$ of \emph{finished} nodes, as follows:
\begin{enumerate}
\item Initially, set $P = (s)$ and $F = \emptyset$, where $s$ is the starting node.
\item Let $P = (u_1, u_2, \dots, u_k)$. If there exists a neighbor $u \in N(u_k) \setminus (F \cup \{u_1, u_2, \dots, u_k\})$, choose any such node and append it to $P$. Otherwise, remove the last node (or \emph{head}) $u_k$ from $P$ and add it to $F$.
\item If $P$ is empty (\ie $P=()$), terminate; otherwise, return to step (2).
\end{enumerate}
We call $P$ the \emph{return path}. Generally speaking, during a DFS, two nodes $u_i$ and $u_j$ in the return path $P=(u_1,u_2,\dots,u_k)$ may be adjacent even if $|i-j|\neq 1$. In particular, the head $u_k$ may have two or more neighbors that belong to $P$. When the agent attempts to simulate DFS using a constant number of colors, this fact poses a significant challenge (Figure \ref{fig:path}). To backtrack from $u_k$ to $u_{k-1}$, the agent must distinguish $u_{k-1}$ from the other neighbors of $u_k$. If we were allowed to assign different colors to the $k-1$ nodes $u_1, u_2, \dots, u_{k-1}$, the agent could easily identify $u_{k-1}$. However, this approach would require $\Omega(|V|)$ colors generally. Thus, there appears to be no straightforward way to implement DFS with a constant number of colors.

\begin{algorithm}[t]
\caption{Semi-DFS}
\label{al:semi-dfs}
$P \gets (s)$,\ $F \gets \emptyset$\tcp*{$s$ is the starting node}
\While{$P$ is not empty (\ie $P\neq ()$)}{
Let $(u_1,u_2,\dots,u_k)=P$, where $k=|P|$.\;
Let $U=\upf$.\;
\uIf{$U \neq \emptyset$}{
Choose an arbitrary node $u \in U$.\; 
$P \gets (u_1,u_2,\dots,u_k,u)$\tcp*{expand path $P$}
}
\Else{
$P \gets (u_1,u_2,\dots,u_{k-1})$\tcp*{remove $u_k$ from path $P$}
$F \gets F \cup \{u_k\}$\;
}
}
\end{algorithm}

\begin{algorithm}[t]
\caption{Graph Exploration Algorithm for General Graphs $\asix=(\csix,\cinit,\xisix)$}
\label{al:asix}
\Colors{
$\csix=\{\cinit,\cpath,\cfin,\cheadone,\cheadtwo,\cneigh\}$
}
\RuleFunction{\\
\begin{align*}
\xisix(\cinit,M) &= 
\begin{cases}
(\cneigh,\cheadone)& \mif \cpath \in M  \hfill (\rlabel{2.1})\\
(\cheadone,\cheadone)& \melif \cheadone \in M \hspace{1.1cm} \hfill (\rlabel{2.2})\\
(\cheadone,\cstay) & \moth \hfill (\rlabel{2.3})
\end{cases}\\
\xisix(\cheadone,M) &= 
\begin{cases}
(\cheadone,\cinit)& \mif \cinit \in M \land \cheadone \notin M \hspace{0.52cm}  \hfill (\rlabel{2.4})\\
(\cheadtwo,\cstay)& \moth \hfill (\rlabel{2.5})\\ 
\end{cases}\\
\xisix(\cheadtwo,M) &= 
\begin{cases}
(\cheadtwo,\cneigh)& \mif \cneigh \in M \hspace{2.07cm} \hfill  (\rlabel{2.6})\\
(\cpath,\cheadone)& \melif \cheadone \in M \hfill (\rlabel{2.7})\\
(\cfin,\cpath) & \melif \cpath \in M \hfill (\rlabel{2.8})\\
(\cfin,\cstop) & \moth \hfill (\rlabel{2.9})\\
\end{cases}\\
\xisix(\cpath,M) &=  (\cheadone,\cstay) \hspace{4.9cm} \hfill (\rlabel{2.10})\\
\xisix(\cneigh,M) & = (\cinit,\cheadtwo)
\hspace{5cm} \hfill (\rlabel{2.11})
\end{align*}
}
\end{algorithm}

Therefore, we modify DFS so that the return path is always
an \emph{induced path}, meaning that no two nonconsecutive
nodes on the path are adjacent.
We refer to this modified procedure as the semi-DFS. As we shall see in Section~\ref{sec:six_color_alg}, this property greatly facilitates graph exploration. The pseudocode for the semi-DFS is given in Algorithm~\ref{al:semi-dfs}. The only difference from the standard DFS arises during path expansion: we extend the path $P=(u_1,u_2,\dots,u_k)$ only if this property is preserved. 
Specifically, let $U(P,F)$ denote the set of neighbors of $u_k$ that have not been finished, do not belong to the current path $P$, and are not adjacent to any earlier node on the path, i.e.,
$U(P,F)=\upf$.
If $U(P,F)\neq\emptyset$, we extend the path; otherwise,
we backtrack from $u_k$ to $u_{k-1}$ if $k>1$, and terminate if $k=1$.

\begin{lemma}
\label{lemma:semi-dfs}
For any graph $G=(V,E)$ and for any starting node $s\in V$, the semi-DFS algorithm eventually terminates, and by the time it terminates, every node in $G$ has been visited (i.e., $F=V$ eventually holds).
\end{lemma}

\begin{proof}
We prove that Algorithm~\ref{al:semi-dfs} terminates and that $F=V$ holds upon termination.
\paragraph*{Termination}
At each iteration of the while loop, either (i) a node is appended to the path $P$, or (ii) a node is removed from $P$ and added to $F$. Since each node can be appended to $P$ at most once and removed from $P$ at most once---and once a node is added to $F$, it is never removed---the total number of iterations is bounded by $2|V|$. Hence, the algorithm terminates after a finite number of iterations.

\paragraph*{Full Visitation}  
Let $X\subseteq V$ be the set $F$ upon termination. Assume, for the sake of contradiction, that $X\subsetneq V$.
Since $G$ is connected, the boundary set $B(X) = \{v \in X \mid \exists u \in V \setminus X: \{u,v\} \in E\}$
is not empty. Let $\vin$ be the node in $B(X)$ that was removed from the return path $P$ most recently, and choose any $\vout \in N(\vin)\setminus X \neq \emptyset$. 
Suppose the return path was $P=(u_1,u_2,\dots,u_k)$ with $\vin=u_k$
at the moment when $\vin$ was removed,
at which $U(P,F) = \upf$ must be empty.
Since $\vout\in N(u_k)$, at least one of the following must hold:
\begin{enumerate}
\item $\vout\in F \subseteq X$,
\item $\vout \in\{u_1,u_2,\dots,u_k\} \subseteq X$, or
\item $\vout$ is adjacent to some $u_i$ with $1\le i\le k-1$.
\end{enumerate}
Cases 1 and 2 contradict $v_{\mathrm{out}}\notin X$. Case 3 yields $u_i\in B(X)$ for some $i<k$, contradicting the choice of $v_{\mathrm{in}}$ as the most recently removed boundary node. Hence $X=V$, as required.
\end{proof}

\begin{figure}[t]
 \centering 
 \includegraphics[width=0.3\textwidth]{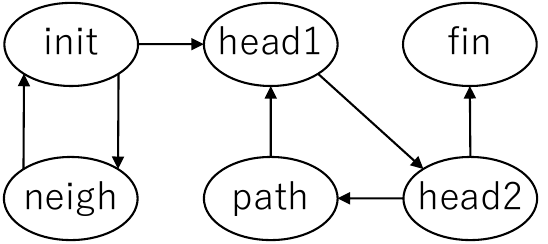}
 \caption{Lifecycle of node colors in algorithm $\asix$} 
 \label{fig:transitions}
\end{figure}

\subsection{How to Simulate Semi-DFS with Six Colors}
\label{sec:six_color_alg}

In this section, we present an algorithm
$\asix=(\csix,\cinit,\xisix)$ that simulates semi-DFS using six
colors: $\cinit$, $\cpath$, $\cfin$, $\cheadone$, $\cheadtwo$,
and $\cneigh$.
The algorithm is defined in Algorithm~\ref{al:asix}.
Although $\xisix(\cfin,M)$ is left unspecified, this causes
no issue: as we show below, the agent is located at a node
colored $\cfin$ only after termination.
Figure~\ref{fig:transitions} illustrates how node colors
change from $\cinit$ to $\cfin$.

Algorithm $\asix$ simulates semi-DFS as specified in
Algorithm~\ref{al:semi-dfs}.
At the beginning of each simulated iteration, nodes colored
$\cfin$ represent the finished nodes in $F$, while nodes
colored $\cpath$ and $\cheadone$ form the return path
$P=(u_1,u_2,\dots,u_k)$.
The head $u_k$ is colored $\cheadone$, and all earlier nodes
on $P$ are colored $\cpath$.
All remaining nodes are colored $\cinit$.
The colors $\cheadtwo$ and $\cneigh$ are used temporarily
when searching for a node in $U(P,F)=\upf$.

We now describe how each iteration is simulated, using
Figures~\ref{fig:expansion} and~\ref{fig:removal} to illustrate
path expansion and backtracking, respectively.
Each simulated iteration consists of \emph{searching mode}
followed by \emph{cleaning mode}, with the head $u_k$ colored
$\cheadone$ and $\cheadtwo$, respectively.
In searching mode, the agent visits neighbors of $u_k$
colored $\cinit$ to find a node in $U(P,F)$.
Such a node belongs to $U(P,F)$ if and only if it has no
neighbor colored $\cpath$.
The agent colors the visited node with $\cheadone$ if it
satisfies this condition, or with $\cneigh$ otherwise,
and returns to $u_k$.
Once a suitable node is found or no neighbor of $u_k$
remains colored $\cinit$, the agent switches to cleaning
mode by recoloring $u_k$ with $\cheadtwo$.
In this mode, it visits each neighbor colored $\cneigh$,
restores its color to $\cinit$, and returns to $u_k$.
After cleaning, if a suitable node was found, the agent
recolors $u_k$ with $\cpath$ and moves to that node,
thereby extending the return path.
Otherwise, it recolors $u_k$ with $\cfin$ and, if $k>1$,
backtracks to $u_{k-1}$ and recolors it with $\cheadone$.
If $k=1$, it instead terminates at the starting node.

The proof of the following theorem explains these steps
in detail and establishes correctness.

\begin{theorem}
\label{theorem:asix}
Algorithm $\asix=(\csix,\cinit,\xisix)$, which uses six colors, solves the exploration problem on any graph $G=(V,E)$.
\end{theorem}

\begin{proof}
We fix a starting node $s\in V$ and consider an arbitrary maximal execution $\calE=\gamma_0,\gamma_1,\dots$ of $\asix$ starting at $s$. To prove correctness, it suffices to show that $\calE$ terminates at $s$ with every node colored $\cfin$, since a node can be recolored only after it has been visited.

We say that a configuration $\gamma=(v,\psi,b)$ of $\asix$ is \emph{regular} if either 
\begin{enumerate}
\item $b=\tr$, or 
\item $b=\fl$ and there exists an induced path $P=(u_1,u_2,\dots,u_k)$ with $k \ge 1$ such that $u_1 = s$, $u_k = v$,
$\psi(u_k) = \cheadone$, $\psi(u_i) = \cpath$ for all $1 \le i < k$, and every node not on $P$ is colored either $\cinit$ or $\cfin$.
\end{enumerate}
We call a configuration \emph{final} if it satisfies the first condition, i.e., $b=\tr$.
If $\gamma$ is regular but not final, the path $P$ is uniquely
determined: its nodes are exactly those colored $\cpath$ or
$\cheadone$, and since they induce a path, their order is uniquely
determined by the starting node $s$.
We denote this path by $P(\gamma)$ and call it the
\emph{return path} of $\gamma$.
We also define $F(\gamma)=\{u\in V : \psi(u)=\cfin\}$.

In the initial configuration $\gamma_0$, each node has color $\cinit$, and the agent is located at $s$.
Therefore, by \rlabel{2.3}, only $s$ is colored $\cheadone$, while all other nodes remain colored $\cinit$, and the agent remains at $s$ in $\gamma_1$. This configuration $\gamma_1$ is regular and corresponds to the initial setting in the semi-DFS, \ie $P(\gamma_1)=(s)$ and $F(\gamma_1)=\emptyset$. In the remainder of this proof, we show that from any regular but non-final configuration $\gamma$, $\asix$ simulates one iteration of the semi-DFS while loop, yielding a new regular configuration.
By Lemma~\ref{lemma:semi-dfs}, it follows that the agent visits every node and terminates at $s$, proving correctness.

\begin{figure}[t]
    \centering
    \includegraphics[width=\textwidth]{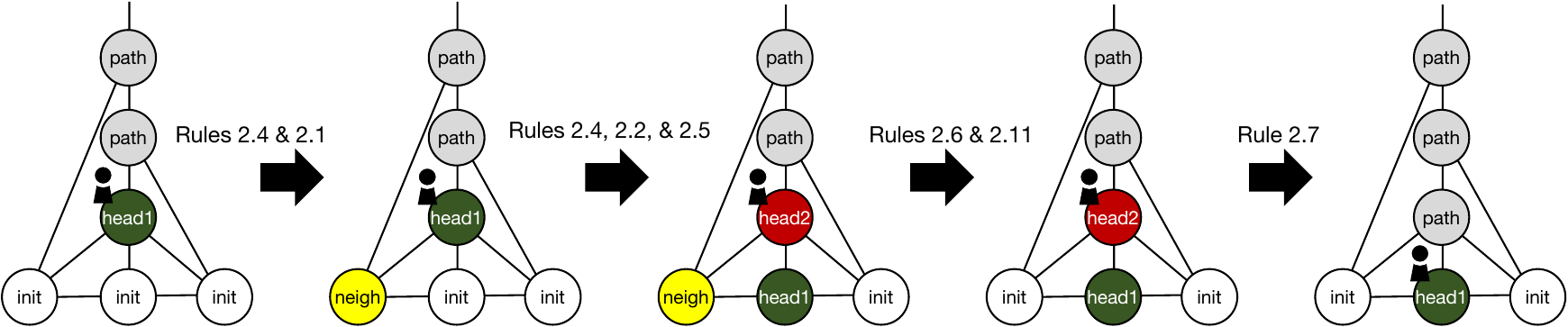}
    \caption{Path expansion in $\asix$.}
    \label{fig:expansion}
    \vspace{0.3cm}
    \includegraphics[width=\textwidth]{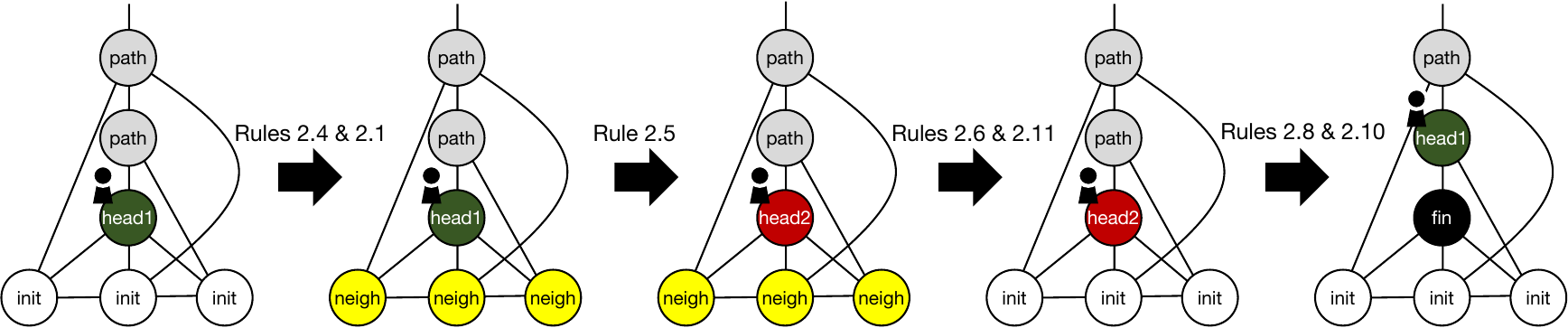}
    \caption{Removal of the head in $\asix$.
    The graph differs from that in Figure~\ref{fig:expansion}
    by the addition of a single edge.}
    \label{fig:removal}
\end{figure}

Let $\gamma=(v,\psi,b)$ be any regular, non-final configuration. In $\gamma$, the agent is located at the head of the return path $P(\gamma)=(u_1,u_2,\dots,u_k)$; that is, the agent is at node $v=u_k$, which is colored $\cheadone$. Our goal is to verify that the agent simulates one iteration of the while loop of the semi-DFS from this configuration: if there exists a node in $U(P(\gamma),F(\gamma))=\upfg{\gamma}$, the agent appends an arbitrary such node to $P$; if $U(P(\gamma),F(\gamma))=\emptyset$ and $k\ge 2$, the agent removes the head $u_k$ from $P$ and backtracks to $u_{k-1}$, recoloring $u_k$ with $\cfin$; otherwise, the agent terminates.

Let $I(\gamma)=\{ u\in N(u_k) \mid \psi(u)=\cinit\}$.
By \rlabel{2.4}, the agent visits each node $u\in I(\gamma)$ one by one. Checking whether $u\in U(P(\gamma),F(\gamma))$ is straightforward: since $u$ is colored $\cinit$, it belongs to $U(P(\gamma),F(\gamma))$ if and only if none of its neighbors is colored $\cpath$. The agent colors $u$ with $\cheadone$ if $u\in U(P(\gamma),F(\gamma))$, and with $\cneigh$ otherwise, after which the agent returns to $u_k$ (\rlabels{2.1 and 2.2}). We now consider two cases.

\paragraph*{Case 1. $U(P,F)\neq\emptyset$}  
In this case, the agent visits the nodes in $I(\gamma)$, recoloring them with $\cneigh$, until it encounters a node $u\in U(P,F)$. Then, the agent colors $u$ with $\cheadone$ and returns to $u_k$. At $u_k$, upon observing that $u$ is colored $\cheadone$, the agent recolors $u_k$ with $\cheadtwo$ (\rlabel{2.5}). By \rlabels{2.6 and 2.11}, the agent then visits all neighbors colored $\cneigh$ and reverts their colors back to $\cinit$. Finally, the agent expands $P$ by recoloring $u_k$ with $\cpath$ and moving to $u$ (\rlabel{2.7}).  

\paragraph*{Case 2. $U(P,F)=\emptyset$}  
In this case, the agent recolors all nodes in $I(\gamma)$ with $\cneigh$, recolors $u_k$ with $\cheadtwo$ (\rlabel{2.5}), and reverts their colors back to $\cinit$ (\rlabels{2.6 and 2.11}). If $k\ge 2$, the agent moves to and recolors $u_{k-1}$ with $\cheadone$ (\rlabels{2.8 and 2.10}); otherwise, it terminates (\rlabel{2.9}). In any case, $u_k$ is recolored with $\cfin$.

In both cases, we have shown that the agent correctly simulates one complete iteration of the semi-DFS while loop, resulting in a new regular configuration.
\end{proof}


\begin{algorithm}[t]
\caption{Graph Exploration Algorithm for Triangle-free Graphs $\atf=(\ctf,\cinit,\xitf)$}
\label{al:triangle-free}
\Colors{
$\ctf=\{\cinit,\cpath,\cneigh,\cfin,\cheadd\}$
}
\RuleFunction{
\vspace{-0.5cm}
\\
\begin{align*}
\xitf(\cinit,M) &= 
\begin{cases}
(\cneigh,\cheadd)& \mif \cheadd \in M \land \cpath \in M \hspace{1.05cm} \hfill (\rlabel{3.1})\\
(\cheadd,\cheadd)& \melif \cheadd \in M \land \cpath \notin M \hfill (\rlabel{3.2})\\
(\cheadd,\cstay) & \moth\ \hfill (\rlabel{3.3})\\
\end{cases}\\
\xitf(\cheadd,M) &= 
\begin{cases}
(\cheadd,\cneigh)& \mif \cheadd \in M \land \cneigh \in M \hspace{0.93cm} \hfill (\rlabel{3.4})\;\;\\
(\cpath,\cheadd)& \melif \cheadd \in M \land \cpath \in M \hfill (\rlabel{3.5})\;\;\\
(\cfin,\cheadd)& \melif \cheadd \in M \hfill (\rlabel{3.6})\;\;\\
(\cheadd,\cinit)& \melif \cpath \in M \land \cinit \in M \hfill (\rlabel{3.7})\;\;\\
(\cheadd,\cpath)& \melif \cpath \in M  \hfill (\rlabel{3.8})\;\;\\
(\cpath,\cinit)& \melif \cinit \in M  \hfill (\rlabel{3.9})\;\;\\
(\cfin,\cstop)& \moth \hfill (\rlabel{3.10})
\end{cases}\\
\xitf(\cpath,M) &=  (\cheadd,\cheadd) \hspace{5.55cm} \hfill (\rlabel{3.11})\\
\xitf(\cneigh,M) & = (\cinit,\cheadd) \hspace{5.75cm} \hfill (\rlabel{3.12})\\
\end{align*}
}
\vspace{-0.5cm}
\end{algorithm}

\begin{figure}[t] 
    \centering
    \includegraphics[width=\textwidth]{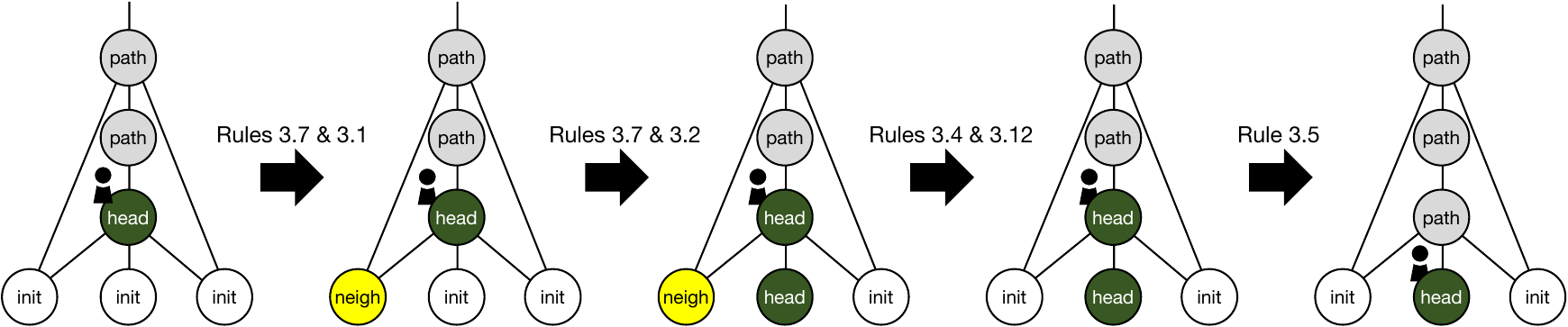}
    \caption{Path expansion in $\atf$}
    \label{fig:expansion_tf}      
    \vspace{0.3cm}
    \includegraphics[width=\textwidth]{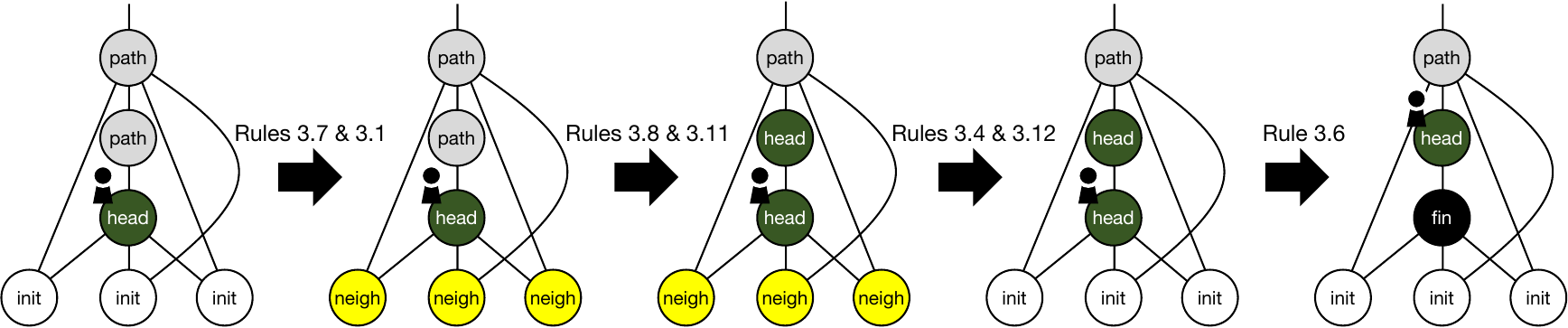}
    \caption{Removal of the head in $\atf$. The graph differs from that in Figure~\ref{fig:expansion_tf}
    by the addition of a single edge.}
    \label{fig:removal_tf}
\end{figure}

\subsection{Triangle-free Case}
When we restrict attention to triangle-free graphs (i.e., graphs with no cycle of length $3$), we can reduce the number of colors by one. Specifically, by merging the two colors $\cheadone$ and $\cheadtwo$ in $\asix$ into a single color $\cheadd$, the algorithm uses only five colors. In this section, we present the resulting algorithm $\atf=(\ctf,\cinit,\xitf)$, whose pseudocode is shown in Algorithm~\ref{al:triangle-free}.

The colors $\cheadone$ and $\cheadtwo$ are used only to distinguish whether the agent is in the \emph{searching mode} or the \emph{cleaning mode}. Since these two colors are merged into a single color $\cheadd$, the agent must determine its current mode in $\atf$ differently. First, consider the case where the head node $u_k$ is not the starting node, i.e., $k\ge 2$ (Figures \ref{fig:expansion_tf} and \ref{fig:removal_tf}). The agent determines its mode as follows: it is in the searching mode if no neighbor of $u_k$ is colored $\cheadd$; otherwise, it is in the cleaning mode.

Suppose the agent is in the searching mode. If $U(P,F)\neq\emptyset$ and the agent finds a node in $U(P,F)$, it colors that node with $\cheadd$ and returns to $u_k$ (\rlabel{3.2}). If $U(P,F)=\emptyset$ and no node in $N(u_k)$ is colored $\cinit$, the agent tentatively backtracks to $u_{k-1}$, colors it with $\cheadd$, and returns to $u_k$ (\rlabels{3.8 and 3.11}). In both cases, exactly one neighbor of $u_k$ is colored $\cheadd$, so the agent correctly identifies the current mode as the cleaning mode.

In the cleaning mode, the agent visits all nodes in $N(u_k)$ colored $\cneigh$, recolors them to $\cinit$, and returns to $u_k$ (\rlabels{3.4 and 3.12}). At this point, the agent can distinguish the two cases: $U(P,F)\neq\emptyset$ if and only if there exists a node colored $\cpath$ in $N(u_k)$. If so, it extends the path (\rlabel{3.5}); otherwise, it removes the head $u_k$ and backtracks to $u_{k-1}$ (\rlabel{3.6}). The triangle-free assumption is required only in the cleaning mode. Without this assumption, when the agent visits a node colored $\cneigh$ by \rlabel{3.4}, it may observe two neighbors colored $\cheadd$. In this case, the agent cannot distinguish $u_k$ from the other $\cheadd$-colored node, and thus may fail to return to $u_k$ by \rlabel{3.12}. This issue does not arise when the graph contains no cycle of length $3$.

Finally, consider the case $k=1$. In this case, the cleaning mode is unnecessary, since all nodes colored $\cinit$ in $N(u_k)$ belong to $U(P,F)$. Hence, the agent extends the path if there exists a neighbor colored $\cinit$ (\rlabels{3.9 and 3.3}); otherwise, it terminates (\rlabel{3.10}).

\begin{theorem}
\label{theorem:triangle_free}
Algorithm $\atf=(\ctf,\cinit,\xitf)$, which uses five colors, solves the exploration problem on any graph $G=(V,E)$ containing no cycle of length $3$.
\end{theorem}

\begin{proof}
We use the same notion of \emph{regular} configuration as in the proof of Theorem~\ref{theorem:asix}, together with the associated notation $P(\gamma)$ and $F(\gamma)$, except that the head of the return path is colored $\cheadd$ instead of $\cheadone$. As before, after one application of \rlabel{3.3} to the initial configuration, we obtain a regular configuration corresponding to the initial state of the semi-DFS.

Thus it suffices to show that, from every regular non-final configuration $\gamma$, algorithm $\atf$ reaches the next regular configuration exactly as one iteration of the semi-DFS on $(P(\gamma),F(\gamma))$. The verification is the same as in the proof of Theorem~\ref{theorem:asix}; the only new point is that the colors $\cheadone$ and $\cheadtwo$ have been merged into a single color $\cheadd$.

Let $P(\gamma)=(u_1,\dots,u_k)$, and let $u_k$ be its head. If $k=1$, then every $\cinit$-colored neighbor of $u_k=s$ belongs to $U(P(\gamma),F(\gamma))$, so \rlabels{3.9 and 3.3} extend the path whenever such a neighbor exists, and \rlabel{3.10} terminates otherwise.

Assume now $k\ge 2$. In the searching phase, exactly as in the proof of Theorem~\ref{theorem:asix}, the agent scans the $\cinit$-colored neighbors of $u_k$, recoloring those outside $U(P(\gamma),F(\gamma))$ with $\cneigh$ and, if it finds one in $U(P(\gamma),F(\gamma))$, recoloring it with $\cheadd$. If no such node exists, the agent recolors $u_{k-1}$ with $\cheadd$ by \rlabels{3.8 and 3.11}. Hence, when the cleaning phase starts, $u_k$ has exactly one neighbor colored $\cheadd$: it is either the newly found node in $U(P(\gamma),F(\gamma))$ or the predecessor $u_{k-1}$.

The triangle-free assumption is used only here. During the cleaning phase, the agent repeatedly visits a node colored $\cneigh$ by \rlabel{3.4} and returns by \rlabel{3.12}. Let $x$ be such a $\cneigh$-colored node, and let $y$ be the unique $\cheadd$-colored neighbor of $u_k$. If $x$ were adjacent to $y$, then $u_kxy$ would form a triangle. Therefore $x$ is adjacent to no $\cheadd$-colored node other than $u_k$, and \rlabel{3.12} necessarily sends the agent back to $u_k$. Consequently, the cleaning phase is unambiguous.

After all nodes colored $\cneigh$ have been reset to $\cinit$, the head $u_k$ sees a neighbor colored $\cpath$ if and only if a node in $U(P(\gamma),F(\gamma))$ was found. Hence \rlabel{3.5} extends the return path in that case, whereas \rlabel{3.6} removes $u_k$ from the return path and colors it $\cfin$ otherwise. Thus $\atf$ simulates one iteration of the semi-DFS from $\gamma$, yielding a new regular configuration.
\end{proof}

\begin{figure}[tbp]
  \begin{minipage}[b]{0.62\columnwidth}
    \centering
    \includegraphics[width=0.8\columnwidth]{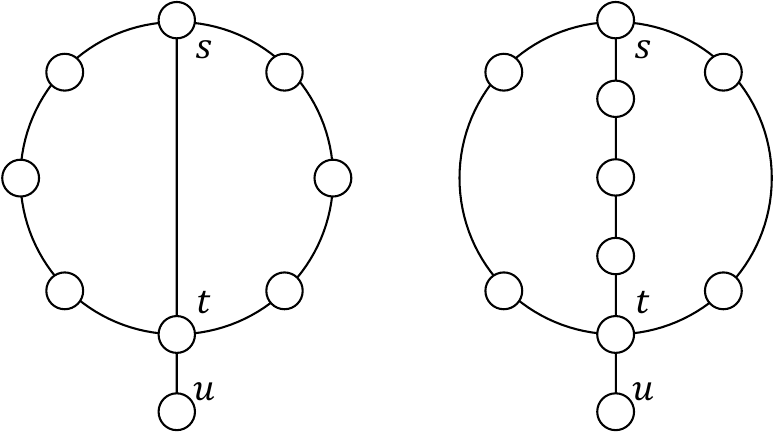}
    \caption{$\varphi$-shaped graphs} 
     \label{fig:phi_graph}    
  \end{minipage}
  \hspace{0.02\columnwidth} 
  \begin{minipage}[b]{0.32\columnwidth}
    \centering
    \includegraphics[width=0.6\columnwidth]{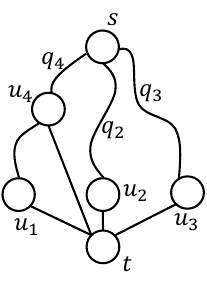}
    \caption{$\varphi$-shaped graph in the proof of Lemma~\ref{lemma:phi-free}}
    \label{fig:phi-free_lemma}
  \end{minipage}
\end{figure}

\section{Exploration on $\varphi$-free Graphs}
\label{sec:afive}

In this section, we show that five colors suffice for exploring $\varphi$-free graphs, defined below.
\begin{definition}[$\varphi$-free graph]
A graph $G=(V,E)$ is said to be \emph{$\varphi$-shaped} if there exist three distinct nodes $s,t,u\in V$ such that:
\begin{itemize}
    \item there exist three internally node-disjoint paths $p_1$, $p_2$, and $p_3$ between $s$ and $t$ that do not contain $u$,
    \item every node in $V\setminus\{u\}$ lies on at least one of $p_1$, $p_2$, or $p_3$, and 
    \item $\{t,u\}\in E$.
\end{itemize}
A graph is said to be \emph{$\varphi$-free} if it does not contain any $\varphi$-shaped subgraph.
\end{definition}
Typical examples of $\varphi$-shaped graphs are shown in Figure~\ref{fig:phi_graph}. The term $\varphi$-free is derived from the shape of the $\varphi$-shaped graphs. Every graph with maximum degree at most three is $\varphi$-free, since in any $\varphi$-shaped graph the node $t$ has degree at least four. Moreover, every cactus is also $\varphi$-free because every $\varphi$-shaped graph violates the cactus condition that no two cycles share more than one node.
The following useful lemma holds.

\begin{lemma}
\label{lemma:phi-free}
A simple, undirected graph $G=(V,E)$ is $\varphi$-free if and only if, for every node $v\in V$ with degree at least four, no connected component of $G-v$ contains three or more neighbors of $v$.
\end{lemma}
\begin{proof}
\textbf{Sufficiency:\ }
Every $\varphi$-shaped graph contains a node $t$ with degree at least four such that three neighbors of $t$ lie in the same connected component of $G-t$ (see Figure~\ref{fig:phi_graph}). 

\textbf{Necessity:\ }
Assume that there exists a node $t\in V$ with degree at least four for which three of its neighbors, say $u_1$, $u_2$, and $u_3$, lie in the same connected component of $G-t$. Then, in that connected component of $G-t$, there is a spanning tree containing $u_1$, $u_2$, and $u_3$. Consequently, there exists a node $s$ in this component such that there are three internally node-disjoint paths $q_1$, $q_2$, and $q_3$ in $G-t$, where each $q_i$ is a path from $s$ to $u_i$. 

Now, by adding the edge $\{t,u_i\}$ to each $q_i$, we obtain three internally node-disjoint paths $p_1$, $p_2$, and $p_3$ from $s$ to $t$ in $G$. Let $u_4$ be another neighbor of $t$ different from $u_1$, $u_2$, and $u_3$. If none of $p_1$, $p_2$, or $p_3$ contains $u_4$, then $G$ contains a $\varphi$-shaped subgraph, as required. Otherwise, assume without loss of generality that $p_1$ contains $u_4$ (see Figure~\ref{fig:phi-free_lemma}). Let $q_4$ be the subpath of $p_1$ connecting $s$ and $u_4$. Then, the three internally node-disjoint paths $p_4=q_4+\{u_4,t\}$, $p_2$, and $p_3$ from $s$ to $t$ (with the edge $\{t,u_i\}$ appended as before) and edge $\{t,u_1\}$ form a $\varphi$-shaped subgraph.
Note that this discussion holds even when $s = u_2$ or $s=u_3$.
\end{proof}

As mentioned at the beginning of Section~\ref{sec:asix}, the authors of~\cite{BFUW23} use eight colors to explore arbitrary graphs. Their approach employs a breadth-first search (BFS) strategy, encoding the distance between the starting node $s$ and each node $v$ modulo $3$ as the \emph{color} of $v$. Specifically, they use the color set
$$\{\cinit, \cfin, g_0, g_1, g_2, r_0, r_1, r_2\},$$
where we have renamed these colors to be consistent with our notation.
Here, exactly two colors are used to encode each distance (modulo $3$) from the starting node, which they refer to as the green and red colors. This multiplicity (\ie having both a green and a red color for each distance) plays a critical role in computing the correct distance (modulo $3$) from the starting node in the BFS manner. Given these distances, the agent has a \emph{sense of direction}: when the agent is located at a node $v$ with distance $i\in\{0,1,2\}$, the neighbors in $N(v)$ with distance $i-1 \bmod{3}$ are the nodes closer to $s$, while those with distance $i+1 \bmod{3}$ are the nodes further away. This sense of direction greatly facilitates the agent's task of visiting all nodes in the graph.

\begin{algorithm}[t]
\caption{Graph Exploration Algorithm for $\varphi$-free Graphs $\afive=(\cfive,\cinit,\xifive)$
(Note: All arithmetic on indices in \(\{0,1,2\}\) is performed modulo 3; the modulo notation is omitted.)
}
\label{al:afive}
\Colors{
$\cfive=\{\cinit,d_0,d_1,d_2,\cfin\}$
}
\Notation{
$f(M) = \text{the unique } i \in \{0,1,2\}
\text{ with } d_{i-1}\in M \land d_{i+1}\notin M$.
(Used only in \rlabel{4.2}, where such an $i$ exists uniquely.)
}
\RuleFunction{\\
\begin{align*}
\xifive(\cinit,M) &= 
\begin{cases}
(\cfin,d_0)& \mif \{d_0,d_1,d_2\}\subseteq M  \hspace{2cm} \hfill (\rlabel{4.1})\\
(d_{f(M)},\cstay)& \melif \exists j \in \{0,1,2\}: d_j \in M \hfill (\rlabel{4.2})\\
(d_0,\cstay)& \moth \hfill (\rlabel{4.3})
\end{cases}\\
\xifive(d_i,M) &= 
\begin{cases}
(d_i,\cinit)& \mif \cinit \in M  \hspace{3.7cm} \hfill (\rlabel{4.4})\\
(d_i,d_{i+1})& \melif d_{i+1} \in M  \hfill (\rlabel{4.5})\\
(\cfin,d_{i-1})& \melif d_{i-1} \in M \hfill (\rlabel{4.6})\\
(\cfin,\cstop)& \moth \hfill (\rlabel{4.7})
\end{cases}\\
& \hspace{3cm} \text{for each }i \in \{0,1,2\}.
\end{align*}
}
\end{algorithm}

In this section, we show that the multiplicity of colors (i.e., the separate green and red colors) can be eliminated when we consider only $\varphi$-free graphs. Our algorithm $\afive=(\cfive,\cinit,\xifive)$ uses a color set of size five:
$$
\cfive=\{\cinit,\cfin,d_0,d_1,d_2\}.
$$
The key idea is to abandon computing the exact distance between each node $v$ and the starting node $s$. Instead, we construct a directed acyclic graph (DAG) on the graph $G=(V,E)$ using the three colors $d_0,d_1,d_2$, which we refer to collectively as \emph{distance colors}.
Specifically, given a configuration $\gamma=(v,\psi,b)$, we define a digraph $D(\gamma)=(V,A)$ as follows:
$$
A = \{(u,v) \in V \times V \mid \{u,v\} \in E,\; \exists i \in \{0,1,2\}: \psi(u)=d_i \land \psi(v)=d_{i+1}\}.
$$
Intuitively, a node $u$ colored $d_i$ has an outgoing arc to a neighbor $v\in N(u)$ if and only if $v$ is colored $d_{i+1}$. Here and throughout, all arithmetic on indices in $\{0,1,2\}$ is performed modulo 3 (the modulo notation is omitted for brevity); for example, $d_{2+1}=d_0$ and $d_{0-1}=d_2$.
When no ambiguity arises, we omit the argument $\gamma$ and simply denote $D(\gamma)$ by $D$.

The definition of $\afive=(\cfive,\cinit,\xifive)$ is provided in Algorithm~\ref{al:afive}. At time step $0$, the agent colors the starting node $s$ with $d_0$ (\rlabel{4.3}). Thereafter, 
the agent moves to a node with the initial color $\cinit$ only from a node with a distance color by \rlabel{4.4}.
Upon arriving at a node colored $\cinit$, say $v \in V$, the agent observes the multiset
$M = \bigsqcup_{u\in N(v)} \{c(u)\}$,
where $c(u)$ denotes the color of a neighbor $u \in N(v)$. The agent then colors $v$ based on this multiset $M$ and the following two principles:
\begin{enumerate}
\item The digraph $D$ must remain acyclic (i.e., no cycles are ever created), and
\item Any node $v$ that is colored $d_i$ for some $i\in\{0,1,2\}$ must be reachable from $s$ in $D$.
\end{enumerate}
Specifically, \rlabel{4.2} enforces these principles:
whenever $M$ contains at least one but not all three distance
colors, the agent assigns $v$ the color $d_{f(M)}$, where
$f(M)$ is the unique index $i\in\{0,1,2\}$ such that
$d_{i-1}\in M$ and $d_{i+1}\notin M$.
The function $f$ is defined only for such multisets $M$,
which suffices because it is used only in \rlabel{4.2}.
The first principle is preserved because no neighbor of $v$ is colored $d_{f(M)+1}$, and the second principle is maintained because there is at least one neighbor colored $d_{f(M)-1}$. Note that $f(M)$ is uniquely determined here since $M$ contains at least one distance color (recall that the agent has moved to $v$ by \rlabel{4.4}) and $\{d_0,d_1,d_2\} \nsubseteq M$. If, on the other hand, $\{d_0,d_1,d_2\} \subseteq M$,
then the agent colors $v$ with $\cfin$ and then moves to an arbitrary neighbor with color $d_0$ (\rlabel{4.1}).

As mentioned above, the starting node $s$ is colored $d_0$. Consequently, in the digraph $D$, any node $v$ with color $d_i$ is at distance $i$ from $s$ (modulo 3). This provides the agent with a sense of direction, as in the eight-color algorithm of~\cite{BFUW23}.
Suppose that the agent is at a node $v$ with distance color $d_i$ and none of its neighbors is colored $\cinit$. Then, if there exists a neighbor with color $d_{i+1}$, the agent moves forward to that neighbor (\rlabel{4.5});
otherwise, it colors $v$ with $\cfin$ and backtracks to
a neighbor colored $d_{i-1}$ (\rlabel{4.6}). The only exception occurs when the agent is at the starting node $s$ (with color $d_0$) and no neighbor is colored $d_1$; in that case, the agent terminates the execution (\rlabel{4.7}).

The following lemma guarantees that the two principles mentioned above are always preserved.

\begin{lemma}
\label{lemma:principles}
During any execution of $\afive$ starting at any node $s \in V$, the following two invariants are maintained: (i) the digraph $D$ is acyclic, and (ii) every node assigned a distance color is reachable from the starting node $s$ in $D$.
\end{lemma}

\begin{proof}
Initially, no node has a distance color and $D$ has no arcs,
so both invariants hold.
Only \rlabels{4.2, 4.3, 4.6, and 4.7}
can assign or remove a distance color,
so it suffices to consider these rules.

When \rlabel{4.2} assigns $d_i$ to a node $v$,
the definition of $f$ ensures that $v$ has an incoming arc
from a node already reachable from $s$ and no outgoing arcs,
preserving both invariants.

\rlabel{4.3} applies only in the first step, at $s$,
since any later visit to a node colored $\cinit$ occurs via
\rlabel{4.4} from a node with a distance color.
It assigns $d_0$ to $s$, leaving $D$ with no arcs
and $s$ as the only node with a distance color.
Thus, both invariants hold.

When \rlabel{4.6} or \rlabel{4.7} recolors a node $v$
from $d_i$ to $\cfin$, no neighbor of $v$ is colored $d_{i+1}$,
so $v$ has no outgoing arc in $D$.
Removing its incident arcs therefore preserves acyclicity
and reachability from $s$ to every remaining node
with a distance color.
\end{proof}

\begin{lemma}
\label{lemma:rule1}
During any execution of $\afive$ on any $\varphi$-free graph, \rlabel{4.1} is applicable only when the agent is at a node of degree exactly 3.
\end{lemma}
\begin{proof}
Let $v$ be a node that the agent colors with $\cfin$ by
applying \rlabel{4.1} during an execution of $\afive$.
For \rlabel{4.1} to apply, all three distance colors
$d_0,d_1,d_2$ must appear among the neighbors of $v$.
Hence, $v$ has degree at least three.

By \rlabels{4.1, 4.2, and 4.3}, upon its first visit to a node,
the agent assigns it either a distance color or $\cfin$,
and never restores its color to $\cinit$.
Thus, \rlabel{4.1} is applied at $v$ on its first visit.
At that time, at least three neighbors of $v$ have distance
colors and have therefore already been visited.
Since the agent visited these neighbors without passing
through $v$, they belong to the same connected component
of $G-v$.
By Lemma~\ref{lemma:phi-free}, $v$ consequently has degree
at most three, and hence exactly three.
\end{proof}

\begin{lemma}
\label{lemma:connected}
Let $F$ denote the set of nodes colored $\cfin$. During any execution $\gamma_0,\gamma_1,\dots$ of $\afive$ starting at any node $s\in V$ on any $\varphi$-free graph, the induced subgraph $G[V\setminus F]$ is always connected whenever $V\setminus F$ is nonempty.
\end{lemma}
\begin{proof}
Let $F(i)$ be the set $F$ in configuration $\gamma_i$. Assume for contradiction that there exists an integer $i\ge 0$ such that $G[V\setminus F(i)]$ is connected, but $G[V\setminus F(i+1)] = G[V\setminus F(i)] - v$ is nonempty and disconnected,
where $v$ is the node at which the agent is located in $\gamma_i$.
We consider two cases, depending on the rule that causes $v$ to be colored $\cfin$.

\textbf{Case 1:} The agent applied \rlabel{4.1} in $\gamma_i$. By Lemma~\ref{lemma:rule1}, $v$ then has exactly three neighbors, and all three must already have been assigned a distance color.

\textbf{Case 2:} The agent applied \rlabel{4.6} or \rlabel{4.7} in $\gamma_i$. In this case, \rlabel{4.4} was not applicable in $\gamma_i$, implying that every neighbor of $v$ in $N(v)\setminus F(i)$ already has a distance color.

In either case, every neighbor of $v$ in $G[V\setminus F(i)]$ has a distance color. By Lemma~\ref{lemma:principles}, all these neighbors remain reachable from the starting node $s$ in the digraph $D$, even after the removal of $v$. Thus, $G[V\setminus F(i)]- v$ must be connected, which contradicts our assumption.
\end{proof}

\begin{theorem}
\label{theorem:afive}
Algorithm $\afive=(\cfive,\cinit,\xifive)$, which uses five colors, solves the exploration problem on any $\varphi$-free graph.
\end{theorem}
\begin{proof}
We first show that the agent eventually terminates. Notice that the agent changes the color of each node at most twice. Consequently, the rules that alter a node's color (\rlabels{4.1, 4.2, 4.3, 4.6, 4.7}) are applied only finitely many times. Moreover, \rlabel{4.4} is applied only finitely many times since each application of \rlabel{4.4} eventually triggers an application of one of \rlabels{4.1, 4.2, 4.3} in the subsequent step. The only remaining rule is \rlabel{4.5}; however, by Lemma~\ref{lemma:principles}, the digraph $D$ remains acyclic throughout the execution, so the agent cannot apply \rlabel{4.5} indefinitely. Thus, the agent must eventually terminate.

Now, consider an execution $\gamma_0,\gamma_1,\dots,\gamma_t$ of $\afive$ starting at some node $s\in V$ on a $\varphi$-free graph, where at step $t-1$ the agent applies \rlabel{4.7} at a node $v$ with distance color $d_i$, thereby coloring $v$ with $\cfin$ and terminating the execution. We claim that $v$ must be the starting node $s$. Otherwise, by Lemma~\ref{lemma:principles}, $v$ would have a neighbor with color $d_{i-1}$, contradicting the condition required to apply \rlabel{4.7}. Hence, the agent terminates at $s$.

Finally, we show that every node is colored $\cfin$ in the final
configuration $\gamma_t$.
Consider the configuration $\gamma_{t-1}$ immediately before
termination, and let $F$ be the set of nodes colored $\cfin$
in this configuration.
Since the agent applies \rlabel{4.7} at $s$, the node $s$ has
no neighbor colored $\cinit$ and no outgoing arc in
$D(\gamma_{t-1})$.
By Lemma~\ref{lemma:principles}, every node with a distance color is reachable
from $s$ in $D(\gamma_{t-1})$, so $s$ is the only such node.
Thus, $s$ is an isolated node in $G[V\setminus F]$.
By Lemma~\ref{lemma:connected}, this graph is connected,
so $V\setminus F=\{s\}$.
The final step recolors $s$ with $\cfin$, leaving every node
colored $\cfin$ in $\gamma_t$.
This completes the proof.
\end{proof}


\begin{figure}[t]
    \centering
    \includegraphics[width=0.35\textwidth]{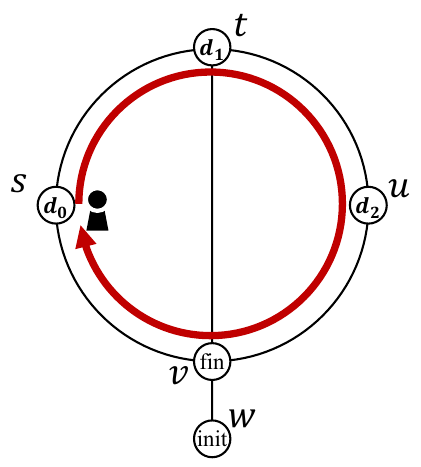}
    \caption{A graph that $\afive$ fails to explore.}
    \label{fig:failure}
\end{figure}

\begin{note}
Our algorithm $\afive$ explores every $\varphi$-free graph.
Although it may also succeed on some graphs that are not
$\varphi$-free, our correctness proof relies on this property,
and the algorithm does not work on arbitrary graphs.
Figure~\ref{fig:failure} gives an example.
Starting at $s$, the agent may follow the path
$s\to t\to u\to v$, assigning $d_0,d_1,d_2$ to $s,t,u$,
respectively.
Upon reaching $v$, it observes all three distance colors
among its neighbors.
It therefore applies \rlabel{4.1}, colors $v$ with $\cfin$,
and returns to $s$, yielding the configuration shown in
the figure.
However, $w$ is still unvisited and is adjacent only to $v$.
Since the agent never moves to a node colored $\cfin$,
it can no longer reach $w$.
It subsequently follows $s\to t\to u\to t\to s$ and
terminates at $s$ without visiting $w$.
Thus, without the $\varphi$-free assumption, marking a node
as finished by \rlabel{4.1} may cut off an unvisited region.
\end{note}

\section{Conclusion}

In this paper, we have advanced the study of graph exploration by a single oblivious agent in the recolorable model, originally introduced by Böckenhauer et al.~\cite{BFUW23}.  Our main contributions are twofold:
\begin{itemize}
  \item  We present a novel exploration algorithm that reduces the color complexity on arbitrary undirected connected graphs from eight to six colors. Moreover, if the graph is triangle-free (i.e., contains no 3-cycles), the color complexity further reduces to five colors.
  \item We identify a broad graph class---\emph{$\varphi$-free} graphs, which includes all graphs of maximum degree at most three and all cacti---and show that five colors suffice for exploring every graph in this class. 
\end{itemize}
Several intriguing open problems remain.  On arbitrary graphs, no nontrivial lower bound on the number of colors is known, so determining whether six is optimal (or can be further decreased) is a compelling challenge.  Moreover, it would be interesting to study how a small amount of internal agent memory---or relaxations such as randomized moves---might impact the color complexity.

\section*{Acknowledgements}
This work is supported by JST FOREST Program JPMJFR226U and JSPS KAKENHI Grant Numbers JP20KK0232, JP25K03078, and JP25K03079.

  \bibliographystyle{alpha} 
  \bibliography{bfuw}
\end{document}